\documentclass[11pt,a4paper]{article}
\usepackage{amssymb,amsmath,amsthm}
\usepackage[english]{babel}
\usepackage{a4wide}
\usepackage{microtype}
\usepackage{hyperref}
\hypersetup{colorlinks=true,citecolor=blue,linkcolor=blue,urlcolor=blue}

\newtheorem{theorem}{Theorem} 
\newtheorem{lemma}[theorem]{Lemma} 
\newtheorem{corollary}[theorem]{Corollary} 
\newtheorem{proposition}[theorem]{Proposition} 
\theoremstyle{definition} 
\newtheorem{remark}[theorem]{Remark} 
\newtheorem{definition}[theorem]{Definition}

\newtheorem*{proposition*}{Proposition} 
\newtheorem*{lemma*}{Lemma} 
\newtheorem*{theorem*}{Theorem} 

\usepackage[ruled]{algorithm2e}

\SetAlFnt{\small}
\SetAlCapFnt{\small}
\SetAlCapNameFnt{\small}
\SetAlCapHSkip{0pt}
\IncMargin{-\parindent}

\DeclareMathOperator{\Z}{\mathbb Z}
\DeclareMathOperator{\N}{\mathbb N}
\DeclareMathOperator{\Q}{\mathbb Q}
\DeclareMathOperator{\supp}{Supp}
\DeclareMathOperator{\vcsp}{VCSP}
\DeclareMathOperator{\dom}{dom}
\DeclareMathOperator{\g}{\Gamma}
\DeclareMathOperator{\ar}{ar}
\DeclareMathOperator{\QQ}{\mathbb Q\cup\{+\infty\}}
\DeclareMathOperator{\blp}{BLP}
\DeclareMathOperator{\aff}{AIP}
\DeclareMathOperator{\avg}{avg}
\DeclareMathOperator{\pvcsp}{PVCSP}
\DeclareMathOperator{\wma}{wMA}

\def\multiset#1#2{\big(\kern-.2em\big(\genfrac{}{}{0pt}{}{#1}{#2}\big)\kern-.2em\big)}
\def\mmultiset#1#2{\left(\kern-.3em\left(\genfrac{}{}{0pt}{}{#1}{#2}\right)\kern-.3em\right)}

\theoremstyle{definition}
\newtheorem{examplex}[theorem]{Example}
\newenvironment{example}
{\pushQED{\qed}\examplex}
{\popQED\endexamplex}
  
\begin{document}

\title{The combined basic LP and affine IP relaxation\\ for promise VCSPs on
infinite domains\thanks{An extended abstract of this work appeared in the
\emph{Proceedings of the 45th International Symposium on Mathematical Foundations of Computer Science
(MFCS'20)}~\cite{vz20:mfcs}. Stanislav \v{Z}ivn\'y was supported by a Royal Society University Research Fellowship. This project has received funding from the European Research Council (ERC) under the European Union's Horizon 2020 research and innovation programme (grant agreement No 714532). The paper reflects only the authors' views and not the views of the ERC or the European Commission. The European Union is not liable for any use that may be made of the information contained therein.}}

\author{
	Caterina Viola\\
	University of Oxford, UK\\
	\texttt{caterina.viola@cs.ox.ac.uk}
	\and
	Stanislav \v{Z}ivn\'{y}\\
	University of Oxford, UK\\
	\texttt{standa.zivny@cs.ox.ac.uk}
}

\date{}
\maketitle

\begin{abstract}

 Convex relaxations have been instrumental in solvability of constraint
  satisfaction problems (CSPs), as well as in the three different
  generalisations of CSPs: valued CSPs, infinite-domain CSPs, and most recently
  promise CSPs. In this work, we extend an existing tractability result to the
  three generalisations of CSPs combined: We give a sufficient condition for
  the combined basic linear programming and affine integer programming
  relaxation for exact solvability of promise valued CSPs over infinite-domains.
  This extends a result of Brakensiek and Guruswami~[SODA'20] for promise
  (non-valued) CSPs (on finite domains).

\end{abstract}

\section{Introduction}

\paragraph{Constraint satisfaction} Constraint satisfaction problems
(CSPs) are a wide class of computational decision problems. An instance of a CSP
is defined by finitely many relations (constraints) that must hold among
finitely many given variables; the computational task is to decide whether it is
possible to find an assignment of labels from a fixed set (the domain) to the
variables so that all the constraints are satisfied. Many problems in computer
science (e.g., from artificial intelligence, scheduling, computational
linguistic, computational biology and verification) can be modelled as CSPs by
choosing an appropriate set of constraints. However, there are many other
problems in which some of the constraints may be violated at a cost or in which
there are satisfying assignments which are preferable to others. These
situations are captured by valued constraint satisfaction problems. 

\paragraph{Valued constraint satisfaction} An instance of a valued constraint
satisfaction problem (VCSP) is defined by finitely many cost functions (valued
constraints) depending on finitely many given variables and a (rational)
threshold; the computational task is to decide whether it is possible to find an
assignment of labels from the domain to the variables so that the  value of the sum of the cost functions is at most the given threshold. In
VCSP instances, cost functions can take on rational or infinite values. VCSPs
not only capture optimisation problems but are also a generalisation of CSPs:
the non-feasibility of an assignment is modelled by allowing the cost
functions to evaluate to $+\infty$. A CSP can thus be seen as a VCSP in which
the cost functions take values in $\{0,+\infty\}$.

\paragraph{Finite domains} In the case in which the domain (i.e., the fixed set
of possible labels for the variables) is a \emph{finite} set the computational
complexity of both CSPs and VCSPs have been completely classified. Moreover,  in
both frameworks a dichotomy theorem holds: every CSP and VCSP is either in P or
is NP-complete, depending on some algebraic condition of the underlying set of
allowed relations and cost functions, respectively. A dichotomy theorem for CSPs
was conjectured by Feder and Vardi~\cite{FederVardi}. The attempt to prove the
conjecture motivated the introduction of the so-called universal algebraic
approach~\cite{JBK} for CSPs, which was later extended to VCSPs
in~\cite{cccjz13:sicomp} and~\cite{Kozik15:icalp}, where an analogue of the complexity dichotomy was
conjectured for VCSPs. The dichotomy conjectures for finite-domain CSPs and
VCSPs inspired an intensive line of research. A complexity classification of
finite-domain VCSPs for sets of cost functions taking finite (rational-only)
values was established in~\cite{tz16:jacm}. A complexity classification of VCSPs
was consequently established in~\cite{Kolmogorov17:sicomp}, assuming a dichotomy
for CSPs, which was proved independently in~\cite{BulatovFVConjecture}
and~\cite{ZhukFVConjecture}. 

\paragraph{Infinite domains} Although most research on CSPs and VCSPs in the
past two decades focused on finite-domain problems, the literature is full of
problems (studied independently of CSPs and VCSPs) that can be modelled as CSPs
or VCSPs \emph{only} if infinite domains are allowed. 
For instance, solvability of linear Diophantine
equations~\cite{GaussElimPoly,KannanBachem} and the model-checking problem for
Kozen’s modal $\mu$-calculus~\cite{KozenModalMu} are examples of problems that
can be modelled as infinite-domain CSPs. Linear Programming, Linear Least Square
Regression~\cite{BoydVandenberghe}, and Minimum Correlation
Clustering~\cite{CorrelationClustering} are examples of problems that can be
modelled as infinite-domain VCSPs. The classes of infinite-domain CSPs and
infinite-domain VCSPs are huge! In fact, \emph{every} computational problem over
a finite alphabet is polynomial-time Turing-equivalent to an infinite-domain
CSP~\cite{BodirskyGrohe}. Therefore, only by focussing on special classes of
infinite-domain CSPs (and VCSPs) is it possible to obtain general complexity
results. There is a rich literature on the computational complexity of special
classes of infinite-domain CSPs, e.g.,
\cite{tcsps-journal,Essentially-convex,Bodirsky15:jacm,Jonsson2016912,Phylo-Complexity,Bodirsky19:jacm,BMPP16,Barto20:sicomp}.

\paragraph{Promise constraint satisfaction} Both infinite-domains CSPs and VCSPs
are extensions of the original (finite-domain) CSPs. Promise constraint
satisfaction problems (PCSPs) are a third, recently introduced extension of
CSPs~\cite{Brakensiek18:soda,PCSPsconf,BBKO19,Ficak19:icalp}. Informally, in a
PCSP the goal is to find an approximately good solution to a problem under the
assumption (the promise) that the problem has a solution. The difference between
CSPs and PCSPs is that in a PCSP instance each constraint comes with two
relations (not necessarily on the same domain), a ``strict'' and a ``weak''
relation. The computational task is then to distinguish between being able to
satisfy all the strict constraints versus not being able to satisfy all the weak
constraints. A CSP can be seen as a PCSP in which the strict and weak
constraints coincide. Perhaps the most well-known example of a PCSP is the
approximate graph colouring problem, in which the task is to distinguish
$k$-colourable graphs from graphs that are not $c$-colourable, for some $c>k$.
(For $c=k$, we get the standard $k$-colouring problem.)
Kazda recently introduced the framework of promise VCSPs on finite domains~\cite{Kazda20}, where he generalised some of the algebraic reductions from (finite-domain) promise CSPs to (finite-domain) promise VCSPs. As far as we are aware, the only other related work on (finite-domain) promise VCSPs is~\cite{Austrin13:toct}.

\paragraph{Convex relaxations} One of the most effective ways to design a
polynomial-time algorithm for solving combinatorial and optimisation problems is
to employ convex relaxations. The idea of convex relaxations is to transform the
original problem to an integer program which is then relaxed to a
polynomial-time solvable convex program~\cite{BoydVandenberghe}, e.g. a linear
program. 
In the context of CSPs, convex relaxations have been studied for robust
solvability~\cite{Kun12:itcs,Dalmau13:toct,BK16,Dalmau19:sicomp}. 
Convex relaxations have been also successfully applied to the study of the three
extensions of CSPs. For VCSPs, characterisations of the
applicability of the basic linear programming
relaxation~\cite{KolmogorovThapperZivny}, constant levels of the Sherali-Adams
linear programming hierarchy~\cite{Thapper17:sicomp}, and a polynomial-size
semidefinite programming relaxation~\cite{tz18:toct} have been provided for
exact solvability. In the PCSP framework, the polynomial-time tractability via a specific
convex relaxation has been characterised for the basic linear programming
relaxation~\cite{PCSPsconf}, affine integer programming relaxation~\cite{PCSPsconf}, and
their combination~\cite{Brakensiek19:soda,Brakensiek20:soda,BrakensiekGWZ20}.
For infinite-domain VCSPs, a sufficient condition has been identified for the
solvability via a combination of the basic linear programming relaxation and an
efficient sampling algorithm (that is, polynomial-time many-one reduction to a
finite-domain VCSP)~\cite{PLVCSPsolvbyLP,CatThesis}. 

\subsection{Contributions}

We initiate the study of convex relaxations for the three generalisations of
CSPs \emph{combined}; that is, convex relaxations for promise valued constraint
satisfaction problems on infinite-domains. We focus on the combined basic linear
programming (BLP) and affine integer programming (AIP) relaxation introduced by
Brakensiek and Guruswami~\cite{Brakensiek20:soda}. This relaxation is stronger
than both the BLP and AIP relaxations individually in the sense that if a class
of promise VCSPs is solved by, say, the BLP relaxation then it is also solved by
the combined relaxation (and the same holds true for the AIP relaxation). The
power of the combined relaxation for (finite-domain) promise CSPs was
established in~\cite{BrakensiekGWZ20}. Rather surprisingly, the combined
relaxation gives an algorithm that solves all tractable (non-promise) CSPs on
Boolean domains, identified in Schaefer's work~\cite{Schaefer78:complexity},
thus giving a unified algorithm.

By extending the argument from~\cite{Brakensiek20:soda}, we establish a
sufficient algebraic condition on the combined relaxation for the solvability of
promise VCSPs in which the domain of the ``weak cost functions'' is possibly
infinite (Theorem~\ref{thm:main}). The proof of this result draws on ideas
introduced in~\cite{Brakensiek20:soda} but requires a non-trivial amount of
technical machinery to make it work in the infinite-domain valued setting. While
our relaxation is inspired by~\cite{Brakensiek20:soda}, it is appropriately
modified to work in the optimisation setting (of valued (P)CSPs). We
remark that the condition we give is known to be necessary already in special
cases of our setting, namely for finite-domain non-valued
PCSPs~\cite{BrakensiekGWZ20}. As an application of our main result, we derive an
algebraic condition under which an \emph{infinite-domain} promise VCSP admitting
an efficient sampling algorithm can be solved in polynomial time using the
combined relaxation (Theorem~\ref{thm:main2}). We emphasise that our main
results (Theorems~\ref{thm:main} and~\ref{thm:main2}) are appreciatively
general, and in particular hold for various special cases of our framework;
e.g., for finite-domain promise VCSPs and infinite-domain promise CSPs.

\paragraph{Approximability of Max-CSPs} PCSPs are approximability problems in
which we require that all constraints should be satisfied, although only in a
weaker sense. Another very natural and well-studied form of relaxation is to try
to maximise the number of satisfied constraints. Convex relaxations have played
a crucial role in this research direction on approximability of (finite-domain)
Max-CSPs, going back to the work of Goemans and
Williamson~\cite{Goemans95:jacm},
e.g.,~\cite{Raghavendra08:everycsp,tulsiani09:stoc,Lee15:stoc,Chan16:jacm-approx,Chan16:jacm,Kothari17:stoc,Ghosh18:toc}.

\section{Preliminaries}
\label{sec:prelims}

Throughout the paper, we denote by $x_i$ the $i$-th component of a tuple $x$. We
denote by $\N$, $\Z$, $\Q$, and $\Q_{\geq0}$ the set of whole numbers, integer
numbers,  rational numbers, and nonnegative rational numbers, respectively. For every $m\in N$, we denote by $\left[m\right]$ the set $\{1,\ldots, m\}\subset\N$. Finally, for every $k \in \Q$ we use the $\left \lceil k \right\rceil$ and $\left \lfloor k \right\rfloor$ to denote the minimum natural number that is at least $k$ and the maximum natural number that is at most $k$, respectively.

\subsection{Valued Constraint Satisfaction Problems}
\label{sect:vcsps}
A \emph{valued structure} $\Gamma$ (over $D$)  consists of 
a signature $\tau$ consisting of function symbols $f$, each equipped with an arity $\ar(f)$; 
a set $D = \dom(\Gamma)$ (the \emph{domain}); 
and, for each $f \in \tau$, a 
\emph{cost function}, i.e.,  a function $f^{\Gamma} \colon D^{\ar(f)} \to {\Q} \cup \{+\infty\}$. 
Here, $+\infty$ is an extra element with the expected properties that for all $c
\in {\Q} \cup \{+\infty\}$, we have
$(+\infty) + c = c + (+\infty) = +\infty$ and $c < +\infty$ for every $c\in\Q$.
Given a valued structure $\g$ with signature $\tau$, for every $f \in \tau$ we define   $\dom(f):=\{t \in D^{\ar(f)}\mid f^{\g}(t)<+\infty\}$.  

Let $\Gamma$ be a valued structure with domain $D$ and signature $\tau$. 
The \emph{valued constraint satisfaction problem} for $\Gamma$, denoted by $\vcsp(\g)$, is the following computational problem.

An \emph{instance}  of $\vcsp(\g)$ is a triple $I:=(V, \phi,u)$
where 
$V$ is a finite set of variables; 
$\phi$ is an expression  of the form
$\sum_{i=1}^{m} f_i(v^i_1,\ldots, v^i_{\ar(f_i)})$,
where $f_1,\dots,f_m \in \tau$ and all the $v^i_j$ are variables from $V$
(each summand is called a $\tau$-term); and
$u$ is a value from $\Q$. 
The task is to decide whether there exists an assignment $s\colon V \to D$, whose \emph{cost}, defined as
\[\phi^{\g}(s(v_1),\ldots, s(v_{\lvert V \rvert})):=\sum_{i=1}^{m} f^\Gamma_i(s(v^i_1),\ldots, s(v^i_{\ar(f_i)}))\]
is finite, and if so, whether there is one whose cost 
is at most $u$.

We remark that, given a valued structure $\g$ over a \emph{finite} signature,
the representation of the structure $\g$ is inessential for computational
complexity as $\g$ is not part of the input.

\subsection{Fractional Homomorphisms and Fractional Polymorphisms}

Let $X$ be a  set. A \emph{discrete probability measure} on $X$ is a map $\mu \colon  \mathcal P(X) \to \left[0,1\right]$ such that 
$\mu (X)=1$ and $\mu$ satisfies the countable additivity property; i.e., for
every countable collection $\{X_n\}_{n \in \N}$ of pairwise disjoint
subsets $X_n \subseteq X$, it holds that $\mu \big(\bigcup_{n \in
\N}X_n\big)=\sum_{n \in \N}\mu(X_n)$.

Given a probability measure $\mu$ on a countable set $X$, we define its
\emph{support} as the set $\supp(\mu):=\{x \in X \mid \mu (\{x\})>0\}.$ In the
reminder, given a probability measure $\mu$ on a set $X$, we use the notation
$\mu(x):=\mu(\{x\})$. The following proposition is a well-known corollary of
countable additivity (proved in Appendix~\ref{app:prelims} for completeness).

\begin{proposition}\label{prop:prob}
Let $\mu$ be a discrete probability measure on a  set $X$. Then  $\supp(\mu)$  is a countable  set. Furthermore, if $X$ is countable then $\sum_{x \in \supp(\mu)}\mu(x)=1$; that is, $\supp(\mu)$ is non-empty.
\end{proposition}

Proposition~\ref{prop:prob} will guarantee that all 
supports in this paper are countable sets.

Let $\mu$ be a discrete probability measure on a countable set $X$ and let $Y$
be a random variable with countably many possible outcomes $y_1,y_2,\ldots$
occurring with probabilities $\mu(x_1),\mu(x_2),\ldots$, respectively. The
expectation of $Y$ associated with $\mu$ is $\mathbb E_{x \sim \mu}[Y]=\sum_{n
\in \N}\mu(x_n)y_n$.

Let $C$ and $D$ be two sets. A map $g \colon D^m \to C$ is called an $m$-ary
\emph{operation}. For any $m \in \N$,  we denote by $C^{D^m}$ the set of all maps $g \colon D^m \to C$. 

Let $\g$ and $\Delta$ be valued structures with the same signature $\tau$ with domains $C$ and $D$, respectively. 
A \emph{fractional homomorphism}~\cite{ThapperZivny2012}  from $\Delta$ to $\g$ is a discrete probability measure $\chi$ with a non-empty support on  $C^D$  such that for every function symbol $\gamma \in \tau$ and tuple $a \in D^{\ar(\gamma)}$, it holds that \begin{equation}\label{eq:frachom}\mathbb E_{h \sim \chi}[\gamma^{\g}(h(a))]=\sum_{h \in \supp(\chi)}\chi(h)\gamma^{\g}(h(a))\leq \gamma^{\Delta}(a),\end{equation}
where the functions $h$ are applied component-wise. We write $\Delta \to_f \g$ to indicate the existence of a fractional homomorphism from $\Delta$ to $\g$.

The following proposition, proved for completeness in Appendix~\ref{app:prelims}, is adapted
from~\cite{ThapperZivny2012}, where it was proved in the case of finite-domain
valued structure, and appears in~\cite{PLVCSPsolvbyLP}, where it was stated for 
valued structures with arbitrary domains and for fractional homomorphisms with
finite supports. 

\begin{proposition} \label{prop:frachom}
	Let $\g$ and $\Delta$ be valued structures over the same signature $\tau$
  with domains $C$ and $D$, respectively. Assume $\Delta \to_f \g$. Let
  $V=\{v_1,\ldots,v_n\}$ be a set of variables and $\phi$ a sum of finitely many
  $\tau$-terms with variables from $V$. For every $u\in \Q$, if there exists an
  assignment $s \colon V \to D$ such that $\phi^{\Delta}(s(v_1, \ldots, s(v_n))
  \leq u$, then there exists an assignment $s'\colon V \to C$ such that
  $\phi^{\g}(s'(v_1), \ldots, s'(v_n)) \leq u$. In particular, it holds that
  $\inf_{C} \phi^{\g} \leq \inf_{D} \phi^{\Delta}$. 
\end{proposition}

Let $\g$ be a valued structure with  domain $C$ and signature $\tau$. 
An $m$-ary \emph{fractional polymorphism} of $\g$ is a discrete probability
measure on $C^{C^m}$ with a non-empty support such that for every $f \in \tau$
and tuples $a^1,\ldots,a^m \in C^{\ar(f)}$ it holds that \[\mathbb E_{g \sim
\omega}[f^{\g}(g(a^1,\ldots,a^m))]=\sum_{g \in C^{C^m}}\omega(g)f^{\g}(g(a^1,\ldots,a^m))\leq \frac{1}{m}\sum_{i=1}^mf^{\g}(a^i)\] (where
$g$ is applied component-wise).

\subsection{Promise VCSPs}
Let $\g$ and $\Delta$ be two valued structures over the same signature $\tau$ with domains $C$ and $D$, respectively. We say that $(\Delta, \g)$ is a \emph{promise valued template} if there exists a fractional homomorphism from $\Delta$ to $\g$.  
Given a promise valued template $(\Delta,\g)$, the \emph{promise valued
constraint satisfaction problem}~\cite{Kazda20} for $(\Delta,\g)$, denoted by  $\pvcsp(\Delta,\g)$, is the following computational problem.

	An \emph{instance} $I$ of $\pvcsp(\Delta, \g)$	is a triple $I:=(V, \phi,u)$
	where 
      $V$ is a finite set of variables; 
      $\phi$ is an expression  of the form
		$\sum_{i=1}^{m} f_i(v^i_1,\ldots, v^i_{\ar(f_i)})$,
		where $f_1,\dots,f_m \in \tau$ and all the $v^i_j$ are variables from $V$; and
      $u$ is a value from $\Q$. 
	
  The task is to output 
      \textsc{yes} if there exists an assignment $s\colon V \to D$ with cost
		\[\phi^{\Delta}(s(v_1),\ldots, s(v_{\lvert V \rvert})):=\sum_{i=1}^{m} f^\Delta_i(s(v^i_1),\ldots, s(v^i_{\ar(f_i)}))\leq u\]
      and output \textsc{no} if \emph{every}  assignment $s'\colon V \to C$ has cost
		\[\phi^{\g}(s'(v_1),\ldots, s'(v_{\lvert V \rvert})):=\sum_{i=1}^{m} f^{\g}_i(s'(v^i_1),\ldots, s'(v^i_{\ar(f_i)}))\nleq u.\]

Note that every valued structure $\g$ is fractionally homomorphic to itself and thus
  $\vcsp(\g)$ is the same as $\pvcsp(\g,\g)$.

Let  $(\Delta,\g)$ be a promise valued template. We remark that if the common signature $\tau$
is finite then the representation of the template is inessential for the computational complexity
of $\pvcsp(\Delta,\g)$ as $(\Delta,\g$) is not part of the input.

Let $e_i^{(m)}\colon D^m \to D$ denote the $m$-ary projection on $D$ onto the
$i$-th coordinate. Let $\mathcal J^{(m)}_D:=\{e_1^{(m)}, \ldots, e_m^{(m)} \}$,
i.e., the set of all $m$-ary projections on $D$.

An $m$-ary \emph{promise fractional polymorphism}\footnote{These are called weighted polymorphisms in~\cite{Kazda20}.} of a promise valued template $(\Delta,\g)$ is a pair $\omega:=(\omega_I,\omega_O)$ where $\omega_O$ is a discrete probability measure on $C^{D^m}$ with a non-empty support and $\omega_I$  is a  discrete probability measure with (finite) support $\supp(\omega_I)=\mathcal J^{(m)}_D$ such that for every $f \in \tau$ and tuples $a^1,\ldots,a^m \in D^{\ar(f)}$ it holds that
\begin{multline}\label{eq:promisefracpol}
\mathbb E_{g \sim \omega_O}[f^{\g}(g(a^1,\ldots,a^m))]
=
\sum_{g \in \supp(\omega)}\omega_O(g)f^{\g}(g(a^1,\ldots,a^m))\\
\leq\quad
\sum_{i=1}^m\omega_I(e^{(m)}_i)f^{\Delta}(a^i)
=
  \mathbb E_{e \sim \omega_I}[f^{\Delta}(e(a^1,\ldots,a^m)).
\end{multline}

\begin{remark}
An $m$-ary fractional polymorphism $\omega$ of a valued structure $\g$ with domain $C$ can be seen as an $m$-ary promise fractional polymorphism $\mu=(\mu_I,\mu_O)$ of $(\g,\g)$ such that $\mu_O=\omega$ and $\mu_I(e^{(m)}_i)=\frac{1}{m}$ for $1\leq i \leq m$.\end{remark}

\subsection{Block-Symmetric Maps}

Let $S_m$ be the symmetric group on $\{1,\ldots,m\}$. An $m$-ary map $g$ is
\emph{fully symmetric} if for every permutation $\pi \in S_m$, we have
$g(x_1,\ldots,x_m)=g(x_{\pi(1)}, \ldots, x_{\pi(m)})$.

An $m$-ary map $g$ is \emph{block-symmetric} if there exists a partition of the
coordinates of $g$ into blocks $B_1\cup \cdots \cup B_k=\left[m\right]$ such
that $g$ is permutation-invariant within each block $B_i$. Let $\mathcal
P_{sym}(g)$ be the set of all partitions into symmetric blocks of $g$. For
$B_1\cup \cdots \cup B_k\in \mathcal P_{sym}(m)$, we define $w(g,B_1\cup \cdots
\cup B_k):=\min_{1\leq i \leq k}\lvert B_i \rvert$ and we define the
\emph{width} of $g$ to be \[w(g):=\max\{w(g,B_1\cup \cdots \cup B_k)\mid
B_1\cup \cdots \cup B_k \in \mathcal P_{sym}(g)\}.\]

Block-symmetric operation with width $1$ are fully symmetric operations.
The next example shows operations that are block-symmetric but not fully
symmetric.

\begin{example}\label{ex:movingaverage1} 

  We consider so-called ``moving averages'' (see,
  e.g.,~\cite{hu2011removal,raudys2013moving}).
	For $k\in \N$, a $k$-ary \emph{$2$-period weighted centred moving average} is the operation $\wma^{({k})}\colon \Q^k \to \Q$ defined by \[\wma^{({k})}(x^1,\ldots,x^{k}):=\frac{1}{3k}\left( \sum_{i=1}^{\lfloor \frac{k}{4}\rfloor}x^i+2 \sum_{i=\lfloor \frac{k}{4}\rfloor+1}^{\lfloor \frac{3}{4}k\rfloor}x^i+\sum_{i=\lfloor \frac{3}{4}k\rfloor+1}^{k}x^i\right).\]
	It can be verified that a $k$-ary $2$-period weighted centred moving average, where $k=2m+1$, is a block-symmetric operation whose symmetric blocks are $B_1:=\{1,\ldots, \lfloor \frac{k}{4}\rfloor, \lfloor\frac{3}{4}k\rfloor+1, \ldots, k\}$ and $B_2:=\{ \lfloor \frac{k}{4}\rfloor+1,\ldots, \lfloor \frac{3}{4}k\rfloor\}$. Observe that $\lvert B_1\rvert=\lfloor \frac{k}{2}\rfloor$ and $\lvert B_2\rvert=\lceil \frac{k}{2}\rceil$.
\end{example}
An $m$-ary fractional polymorphism $\omega$ of a valued structure $\g$ is \emph{block-symmetric} if there exists a partition of the coordinates of $g$ into blocks $B_1\cup \cdots \cup B_k=\left[m\right]$ such that every operation in $\supp(\omega)$ is permutation-invariant within each coordinate block $B_i$. 

\begin{example}\label{ex:movingaverage2} 

A $k$-ary \emph{average} is a map $A:\Q^k\to\Q$ such that, for every $x\in\Q$,
  $A(x,\ldots,x)=x$ (i.e., $A$ is idempotent) and, for every
  $x^1,\ldots,x^k\in\Q$, we have $\min(x^1,\ldots,x^k)\leq A(x^1,\ldots,x^k)\leq
  \max(x^1,\ldots,x^k)$, and for every $\lambda \in \Q$ it holds $\lambda
  A(x^1,\ldots,x^k)=A(\lambda x^1,\ldots,\lambda x^k)$. Examples of averages
  include the arithmetic average $\avg^{(k)}$ defined by
  $\avg^{(k)}(x^1,\ldots,x^k)=\frac{1}{k}(x^1+\ldots+x^k)$ and the $2$-period
  centred moving average  (cf.~Example \ref{ex:movingaverage1}).

Let $A_1$ and $A_2$ be two different averages. A function $f\colon \Q \to \QQ$
  is called \emph{$(A_1,A_2)$-convex}\footnote{In~\cite{ANDERSON20071294,
  aumann1933konvexe}, the notion of $(A_1,A_2)$-convexity is defined for
  $A_1,A_2$ symmetric averages. An average $A$ is symmetric if for every $k \in
  \N$, and every $x^1, \ldots, x^k \in \Q$, it holds
  $A(x^1,\ldots,x^k)=A(x^{\pi(1)}, \ldots, x^{\pi(k)})$ for every $\pi \in
  S_k$.} if for every $k \in \N$ and every $x^1, \ldots, x^k \in \Q^n$ it holds
  $f(A_1(x^1,\ldots,x^k))\leq A_2(f(x^1),\ldots,f(x^k))$.

 A valued structure with  domain $\Q$  containing only $(\wma, \avg)$-convex
  functions has, for every $m \in \N$, a $(2m+1)$-ary block-symmetric fractional polymorphism $\omega^{({2m+1})}$ defined by 
$\omega^{({2m+1})}(g)=1$ if $g=\wma^{({2m+1})}$ and $0$ otherwise.
\end{example}

Given a promise valued template  $(\Delta,\g)$, an $m$-ary promise fractional polymorphism  $\omega=(\omega_I,\omega_O)$ of $(\Delta,\g)$ is \emph{block-symmetric} if \begin{itemize}
\item  there exists a partition of $\left[m\right]$ into blocks $B_1\cup \cdots \cup B_k$ such that every map in $\supp(\omega_O)$ is s permutation-invariant within each coordinate block $B_i$, and 
\item $\sum_{i \in B_j}\omega_I(e^{(m)}_i)=\frac{\lvert B_j\rvert}{m}$ for every $j \in \{1,\ldots, k\}$.
\end{itemize}

The proof of the following lemma can be found in
Appendix~\ref{app:prelims}. 

\begin{lemma}\label{lemma:weightsyminput}
	Let $(\Delta,\g)$ be a promise valued template and let $m \in \N$. If $\omega=(\omega_I, \omega_O)$ is an $m$-ary block symmetric promise fractional polymorphism of $(\Delta,\g)$, then also $\omega'=(\omega'_I, \omega_O)$, where $\omega'_I(e^{(m)}_i)=\frac{1}{m}$ for $1\leq i \leq m$, is an $m$-ary block-symmetric promise fractional polymorphism of $(\Delta,\g)$.
\end{lemma}

In view of Lemma~\ref{lemma:weightsyminput}, we will assume without loss of generality  that any $m$-ary
block-symmetric promise fractional polymorphism $\omega=(\omega_I, \omega_O)$ is
such that $\omega_I$ assign $\frac{1}{m}$ to each $m$-ary projection on the
domain of $\Delta$ and we will identify $\omega$ with $\omega_O$.

\subsection{The Basic Linear Programming Relaxation}
\label{subsec:blp}

Every VCSP over a finite domain has a natural linear programming relaxation.
Let $\Delta$ be a valued structure with finite domain $D$ and signature $\tau$. Let $I$ be an instance of VCSP($\Delta$) with set of variables $V=\{x_1,\ldots,x_d\}$, objective function 
$\phi(x_1,\ldots,x_d)=\sum_{ j \in J} f_j(x_1^j,\ldots,x_{n_j}^j)$,
with $f_j \in\tau,\; x^j=(x_1^j,\ldots,x_{n_j}^j)  \in V^{n_j}  \text{, for all
} j \in J$ (the set $J$ is finite and indexing the cost functions that are
summands of $\phi$), and a threshold $u \in \Q$.\footnote{Note that the BLP relaxation does not depend on the threshold $u$.}
Define the sets of variables as follows:
$W_1:=\{\lambda_{j}(t)\mid j \in J \text{ and }t  \in D^{n_j}\}$,
$W_2:=\{\mu_{x_i}(a) \mid x_i \in V \text{ and } a \in D\}$,
and $W:=W_1\cup W_2$.
Then the \emph{basic linear programming} (BLP) relaxation  associated to  $I$ (see~\cite{ThapperZivny2012},
\cite{KolmogorovThapperZivny}, and references therein) is a linear program with
variables  $W$ and  is defined in Figure~\ref{fig:blp}.

\begin{figure}[thb]
\fbox{\parbox{0.98\textwidth}{
\[
\blp(I, \Delta):=\min{ \sum_{j \in J} \sum_{ t \in D^{n_j}} \lambda_{j}(t) f^{\Delta}_j(t)} \]
\vskip -.9\baselineskip subject to 
\begin{align*}
\sum_{t \in D^{n_j}: t_{\ell}=a}\lambda_j(t)=\mu_{x_{\ell}^j}(a) &\quad & \text{for all } j \in J \text{, } \ell \in \{1,\ldots,n_j\}\text{, } a \in D,\\ 
\sum_{ a \in D}\mu_{x_i}(a)=1 &\quad &\text{for all } x_i \in V,\\
\lambda_j(t)=0  &\quad & \text{for all } j \in J\text{, } t \notin \dom(f_j),\\
0\leq \lambda_j(t), \mu_{x_i}(a) \leq 1 &\quad & \text{for all } \lambda_j(t) \in W_1,\; \mu_{x_i}(a) \in W_2.
\end{align*}
}}
\caption{BLP}
\label{fig:blp} 
\end{figure}
We remark that a solution to the BLP also satisfies the constraints $\sum_{t \in D^{n^j}}\lambda_j(t)=1$ for all $j \in J$.
If there is no feasible solution to the BLP then $\blp(I,\Delta)=+\infty$. 
For a finite-domain VCSP instance, the corresponding BLP relaxation can be computed in polynomial time. 

We note that it is not difficult to lift the existing results characterising the
power of BLP (in terms of fully symmetric operations) for CSPs~\cite{Kun12:itcs},
VCSPs~\cite{KolmogorovThapperZivny}, promise CSPs~\cite{PCSPsconf,BBKO19}, and
infinite-domain VCSPs~\cite{PLVCSPsolvbyLP,CatThesis} to our setting of promise
VCSPs with infinite domains, cf.~Theorem~\ref{thm:fullysymfpol} in
Appendix~\ref{app:blp}. Our focus, however, is on the stronger combined
relaxation presented in Section~\ref{sec:combined}.

\subsection{The Affine Integer Programming Relaxation}

Let $\Delta$ be a valued structure with finite domain $D$ and signature $\tau$. Let $I$ be an instance of VCSP($\Delta$) with set of variables $V=\{x_1,\ldots,x_d\}$, and objective function 
$\phi(x_1,\ldots,x_d)=\sum_{ j \in J} f_j(x_1^j,\ldots,x_{n_j}^j)$,
with $ f_j \in\tau,\; x^j=(x_1^j,\ldots,x_{n_j}^j)  \in V^{n_j}  \text{, for all
} j \in J$ (the set $J$ is finite and indexing the cost functions that are
summands of $\phi$),  and a threshold $u \in \Q$.\footnote{Note that the AIP relaxation does not depend on the threshold $u$.}
Define the sets of variables as follows:
$R_1:=\{q_{j}(t)\mid j \in J \text{ and }t \in D^{n_j}\}$,
$R_2:=\{r_{x_i}(a) \mid x_i \in V \text{ and } a \in D\}$,
and $R:=R_1\cup R_2$.
Then the \emph{affine integer programming} (AIP) relaxation associated to
$I$~\cite{Brakensiek19:soda,Brakensiek20:soda} is an integer program with
variables  $R$ and  is defined in Figure~\ref{fig:aip}.

\begin{figure}[tbh]
\fbox{\parbox{0.98\textwidth}{
\[
\aff(I, \Delta):=\min{\sum_{j \in J} \sum_{ t \in D^{n_j}} q_{j}(t) f^{\Delta}_j(t)} \]
\vskip -.9\baselineskip subject to 
\begin{align*}
\sum_{t \in D^{n_j}: t_{\ell}=a}q_j(t)=r_{x_{\ell}^j}(a) &\quad & \text{for all } j \in J \text{, } \ell \in \{1,\ldots,n_j\}\text{, } a \in D,\\ 
\sum_{ a \in D}r_{x_i}(a)=1 &\quad &\text{for all } x_i \in V,\\
q_j(t)=0  &\quad & \text{for all } j \in J\text{, } t \notin \dom(f_j),\\
q_j(t), r_{x_i}(a) \in \Z &\quad & \text{for all } q_j(t) \in R_1,\; r_{x_i}(a) \in R_2.
\end{align*}
}}
\caption{AIP}\label{fig:aip}
\end{figure}

We remark that a solution to the AIP also satisfies the constraints $\sum_{t \in
D^{n^j}}\lambda_q(t)=1$ for all $j \in J$. If there is no feasible solution to
the AIP then $\aff(I,\Delta)=+\infty$. For a finite-domain VCSP instance, the
corresponding AIP relaxation can be computed in polynomial time. Since the feasibility version of AIP can be solved in polynomial time~\cite{KannanBachem, Brakensiek19:soda}, (the optimisation version of) AIP can be solved in (oracle) polynomial time using an oracle for the feasibility version of the problem (see~\cite[Theorem~6.4.9]{GLS93}).

\section{The Combined BLP and AIP Relaxation for PVCSPs}
\label{sec:combined}

Let $(\Delta,\g)$ be a promise valued template such that the domain of $\Delta$ is a \emph{finite} set.
We may solve $\pvcsp(\Delta,\g)$ by using a \emph{combination} of the BLP
relaxation and the AIP relaxation of $\Delta$, as proposed (for finite-domain
promise non-valued) CSPs in~\cite{Brakensiek20:soda}, appropriately modified to
the valued setting.

To describe such an algorithm, we need the following definition.

\begin{definition}\label{def:refinement}
Let $\Delta$ be a valued structure with finite domain $D$ and signature $\tau$.
Let us consider an instance $I:=(V,\phi,u)$  of $\vcsp(\Delta)$ such that
$\phi(x_1,\ldots,x_d)=\sum_{ j \in J} f_j(x_1^j,\ldots,x_{n_j}^j)$. Assume that $\blp(I,\Delta)\leq u$. 
We define $(\lambda^\star,\mu^\star)$ as follows. \begin{itemize}
\item If there exists a relative interior point of the rational feasibility polytope of $\blp(I, \Delta)$ with cost at most $u$,\footnote{There is a polynomial-time algorithm~\cite{GLS93, Brakensiek20:soda} that decides the existence of a relative interior point in the rational feasibility polytope of $\blp(I, \Delta)$ with cost at most $u$ and, in the case it exists, finds it.} then $(\lambda^\star,\mu^\star)$ is such a point;
\item otherwise, $(\lambda^\star,\mu^\star)$ is defined to be a point from the relative interior of the optimal polytope of  $\blp(I,\Delta)$.\footnote{Such a point can be found in polynomial time by applying the algorithm in~\cite{GLS93, Brakensiek20:soda} to the feasibility linear program defined by adding to the constraints defining the feasibility polytope of $\blp(I, \Delta)$ the additional constraint $\sum_{j \in J} \sum_{ t \in D^{n_j}} \lambda_{j}(t) f^{\Delta}_j(t)=u$.}     
\end{itemize}
The \emph{refinement of $\aff(I,\Delta)$ with respect to $(\lambda^\star,\mu^\star)$}  is the integer program  $\aff^\star(I,\Delta)$ obtained by adding to $\aff(I,\Delta)$ the constraints
\begin{align*}
&q_j(t)=0 \quad& \text{ for every }  j \in J, t \in D^{n_j} &\text{ such that } \lambda^\star_{j}(t)=0,\\
&r_{x_i}(a)=0 \quad& \text{ for every }  x_i \in V, a \in D &\text{ such that } \mu^\star_{x_i}(a)=0.
\end{align*}
\end{definition}

\begin{algorithm}[tbh] 
	\SetAlgoNoLine
	\KwIn{\\\qquad $I:=(V,\phi,u)$, a valid instance of $\pvcsp(\Delta,\g)$} 
  \KwOut{\\\qquad\textsc{yes} if there exists an assignment $s \colon V\to \dom(\Delta)$ such that $\phi^{\Delta}(s(x_1),\ldots,s(x_{\lvert V\rvert }))\leq u$\\\qquad \textsc{no} if there is no assignment $s \colon V\to \dom(\g)$ such that $\phi^{\Gamma}(s(x_1),\ldots,s(x_{\lvert V\rvert }))\leq u$}
  \medskip
	$\blp(I, \Delta)$\;
  \eIf{$\blp(I,\Delta)\nleq u$}{output \textsc{no}\;
	}{
	    $(\lambda^\star,\mu^\star)$,  as in Definition~\ref{def:refinement}\; 
		$\aff^\star(I,\Delta):=\text{refinement of }\aff(I,\Delta)$ with respect to $(\lambda^\star,\mu^\star)$, as in Definition~\ref{def:refinement}\;
    \eIf{$\aff^\star(I,\Delta)\nleq u$}{output \textsc{no}\;
    }{output \textsc{yes}\;}     }
	
	\caption{The combined BLP $+$ AIP Relaxation Algorithm for $\pvcsp(\Delta,\g)$}
	\label{fig:alg}
\end{algorithm}

As our main result, we now present a sufficient condition under which Algorithm \ref{fig:alg}
correctly solves $\pvcsp(\Delta,\g)$. 

\begin{theorem}\label{thm:main}
Let  $(\Delta,\g)$ be a promise valued template such that  $\Delta$ has a \emph{finite domain}. Assume that for all $L \in \N$ there exists a block-symmetric promise fractional polymorphism of $(\Delta,\g)$ with arity $2L+1$ having two symmetric blocks of size $L+1$ and $L$, respectively. Then Algorithm~\ref{fig:alg}  correctly solves $\pvcsp(\Delta,\g)$ (in polynomial time). 
\end{theorem}

Note that in Theorem~\ref{thm:main} the domain of the valued structure $\g$ can be finite or (countably) infinite. 

To prove Theorem~\ref{thm:main} we need to use a preliminary lemma and the notion of a \emph{bimultiset-structures}.
Let $\Delta$ be a valued $\tau$-structure with domain $D$, let $L \in \N$, and let $B_1\cup B_2$ any partition of $\left[2L+1\right]$ such that $\lvert B_1\rvert=L+1$ and $\lvert B_2\rvert=L$.  The \emph{bimultiset-structure} $\mathcal B^{2L+1}_{B_1,B_2}(\Delta)$   is the valued structure with domain $\multiset{D}{L+1}\times\multiset{D}{L}$ i.e., the set whose elements $(\alpha,\beta)$ are pairs of multisets of elements from $D$ of size $L+1$ and of size $L$, respectively. For every $k$-ary function symbol $f\in \tau$, and $(\alpha_1,\beta_1) \ldots,(\alpha_k,\beta_k) \in
\multiset{D}{L+1}\times\multiset{D}{L}$ the function $f^{\mathcal B^{2L+1}_{B_1,B_2}(\Delta)}$ is defined as follows
\begin{align*}&f^{\mathcal B^{2L+1}_{B_1,B_2}(\Delta)}((\alpha_1,\beta_1)\ldots, (\alpha_k,\beta_k)):=\frac{1}{2L+1}\min_{\substack{t^1,\ldots,t^k \in D^{2L+1}:\\ \{t^{\ell}\}_{B_1}=\alpha_{\ell}, \{t^{\ell}\}_{B_2}=\beta_{\ell}}}\sum_{i=1}^{2L+1}f^{\Delta}(t^1_i,\ldots,t^k_i),\end{align*}
where $\{t^{\ell}\}_{B_1}$ denotes the multiset $\{t^{\ell}_i \mid i \in B_1\}$ and  $\{t^{\ell}\}_{B_2}$ denotes the multiset $\{t^{\ell}_i \mid i \in B_2\}$. 

\begin{lemma}\label{lemma:fhomhalfsym2}
		Let  $(\Delta,\g)$ be a promise valued template such that  $\Delta$ has a
    \emph{finite domain}. Let $L\in \N$ and assume that $(\Delta,\g)$ has a block-symmetric promise fractional polymorphism of arity $2L+1$ with two symmetric blocks $B_1$ and $B_2$ of size $L+1$ and $L$, respectively. Then $\mathcal B^{2L+1}_{B_1,B_2}(\Delta)$ is fractionally homomorphic to $\g$. 
\end{lemma}

\begin{proof}  Let $C$ be the (possibly infinite) domain of $\g$, let $D$ be the finite domain of $\Delta$, and let $\tau$ be the common signature of $\g$ and $\Delta$.
  Let  $\omega$ be the $(2L+1)$-ary block-symmetric promise fractional polymorphism of $(\Delta,\g)$ with  symmetric blocks $B_1$ and $B_2$ of size $L+1$ and $L$, respectively. For every $g \in \supp(\omega)\subseteq C^{D^{2L+1}}$  we define $\tilde g\colon\multiset{D}{L+1}\times\multiset{D}{L}\to C$ by setting, for every $\alpha=\{\xi^j \mid j \in B_1\}\in \multiset{D}{L+1}$ and every $\beta=\{\xi^j \mid j \in B_2\}\in \multiset{D}{L}$, \[\tilde{g}((\alpha,\beta))=g(\xi^1,\xi^2,\ldots,\xi^{2L+1}).\] Observe that $\tilde{g}$ is well defined as $g$ is block-symmetric with symmetric blocks $B_1$ and $B_2$ (the order of the  coordinates from a same block does not matter). We define the discrete probability measure $\chi$ on $C^{\multiset{D}{L+1}\times\multiset{D}{L}}$ as  follows  \[\chi(Y)= \sum_{\substack{g \in \supp(\omega):\\  \tilde{g} \in Y}} \omega(g),\qquad\text{ for every } Y \subseteq \mmultiset{D}{L+1}\times\mmultiset{D}{L}.\]
	
	Observe that $\chi$ satisfies the countable additivity property, since
  $\omega$ does. Furthermore, $\supp(\chi)=\{h \in C^{D^{2L+1}} \mid  h=
  \tilde{g} \text{ for some } g \in \supp(\omega)\}=\{\tilde{g} \in C^{D^{(2L+1)}}
  \mid  g \in \supp(\omega)\}$ is countable by Proposition~\ref{prop:prob} and  it holds that
	\begin{align*}\sum_{h \in \supp(\chi)}\chi(h)= \sum_{g \in \supp(\omega)} \omega(g)=1.\end{align*}
	We claim that $\chi$ is a fractional homomorphism from ${\mathcal B^{2L+1}_{B_1,B_2}(\Delta})$ to $\g$.
	Indeed, for every $f \in \tau$ and every  tuple $((\alpha_1,\beta_1)\ldots,
  (\alpha_k,\beta_k))\in\left(\multiset{D}{L+1}\times\multiset{D}{L}\right)^k$ with $\alpha_i:=\{\xi^j_i\mid j \in B_1\}$, $\beta_i:=\{\xi^j_i\mid j \in B_2\}$, and $k:=\ar(f)$, it holds that 
	{ \begin{align}\nonumber
			& \sum_{h \in C^{\multiset{D}{L+1}\times\multiset{D}{L}}}\chi(h)f^{\g}(h((\alpha_1,\beta_1))\ldots,h((\alpha_k,\beta_k)))\\
      \nonumber= & \sum_{g \in \supp(\omega)}\omega(g) f^{\g}(g(\xi_1^1,\ldots,\xi_1^{2L+1}),\ldots,g(\xi_k^1,\ldots,\xi_k^{2L+1}))\\\label{eq:frhom12}
			=& \sum_{g \in \supp(\omega)}\omega(g)f^{\g}(g(\xi_1^{\pi_1(1)},\ldots,\xi_1^{\pi_1({2L+1})}),\ldots,g(\xi_k^{\pi_k(1)},\ldots,\xi_k^{\pi_k({2L+1})}))\\\label{eq:frhom22}
			\leq &\; \frac{1}{2L+1}\sum_{i=1}^{2L+1} f^{\Delta} (\xi^{\pi_1(i)}_1, \ldots, \xi^{\pi_k(i)}_k)
	\end{align}}
	for every $\pi_1, \ldots,\pi_k \in S_{2L+1}$ such that $\pi_i$ is permutation-invariant within $B_1$ and within $B_2$ for $1\leq i \leq k$. 
	Equality~(\ref{eq:frhom12}) holds because the maps $g \in \supp(\omega)$ are block-symmetric with symmetric blocks $B_1$ and $B_2$. Inequality~(\ref{eq:frhom22}) holds because $\omega$ is a promise fractional polymorphism of $(\Delta,\g)$.
	Then, in particular, we obtain  \begin{align*}
		&\sum_{h \in C^{\multiset{D}{L+1}\times\multiset{D}{L}}}\chi(h)f^{\g}(h((\alpha_1,\beta_1)\ldots,(\alpha_k,\beta_k)))\\\leq &   \frac{1}{2L+1}\min_{\substack{t^1,\ldots,t^k \in D^{2L+1}:\\ \{t^{\ell}\}_{B_1}=\alpha_{\ell},\{t^{\ell}\}_{B_2}=\beta_{\ell}}}\sum_{i=1}^{2L+1}f^{\Delta} (t_{i}^1, \ldots t_{i}^k)= f^{\mathcal{B}^{2L+1}_{B_1,B_2}(\Delta)}((\alpha_1,\beta_1),\ldots,(\alpha_k,\beta_k)).\qedhere
\end{align*}\end{proof}

\begin{remark}\label{rem:converse}
Although we do not need it for our main result, we note that the converse of Lemma~\ref{lemma:fhomhalfsym2} holds true. More precisely, let $(\Delta, \g)$ be a promise valued template such that $\Delta$ has a finite domain. Let $L\in \mathbb N$ and let $B_1\cup B_2$ be a partition of $[2L+1]$ such that $\lvert B_1\rvert =L+1$ and $\lvert B_2\rvert =L$. If $\mathcal{B}^{2L+1}_{B_1,B_2}(\Delta)$ is fractionally homomorphic to $\g$, then $(\Delta, \g)$ has a block-symmetric promise fractional polymorphism of arity $2L+1$ with symmetric blocks $B_1$ and $B_2$. 
 To show this we reason as in the proof of \cite[Lemma~2.2]{ThapperZivny2012}. Let $D$ be the domain of $\Delta$, and let $\chi$ be a fractional homomorphism from $\mathcal{B}^{2L+1}_{B_1,B_2}(\Delta)$ to $\g$. Let us define $h\colon D^{2L+1}\to \multiset{D}{L+1}\times\multiset{D}{L}$ such that it maps a tuple $a=(a_1,\ldots,a_m)\in D^{2L+1}$ to the bimultiset $(\{a\}_{B_1},\{a\}_{B_2})$. Then, it is easily verified that the discrete probability measure $\omega$ on $C^{D^m}$ such that \[\omega(g')=
\sum_{\substack{g \in \supp(\chi) \colon \\ g \circ h=g'}}\omega(g)\] is the desired promise fractional polymorphism.
\end{remark}

\begin{proof}[Proof of Theorem~\ref{thm:main}]
	Let $C$ be the (possibly infinite) domain of $\g$ and let $D$ be the finite domain of $\Delta$. Let $\tau$ be the common signature of $\Delta$ and $\g$.	
	Let $I$ be an instance of $\pvcsp(\Delta,\g)$ with variables $V=\{x_1, \ldots,x_n\}$, objective function $\phi(x_1,\ldots,x_n)=\sum_{j \in J}\gamma_j(x^j)$ where $J$ is a finite set of indices, $\gamma_j \in \g$, and $x^j \in V^{\ar(j)}$, and threshold $u$. 
	
Assume that $\min_D \phi^{\Delta}\leq u$. Our goal is to show that Algorithm~\ref{fig:alg} outputs \textsc{yes}. Since $\min_D \phi^{\Delta}\leq u$ we have $\blp(I, \Delta)\leq u$ and in particular we have that either $\blp(I, \Delta)< u$, which by linearity implies the existence of a relative interior point in the feasibility polytope of $\blp(I, \Delta)$ with value at most $u$; or  $\blp(I, \Delta)= u=\min_D \phi^{\Delta}$. In the first case, each coordinate of $(\lambda^\star,\mu^\star)$  is positive if and only if the same coordinate is positive at some point in the feasibility polytope of the BLP. Therefore, the feasibility lattice of $\aff^\star(I,\Delta)$ includes every possible assignment which is in the support of some feasible
solution to $\blp(I,\Delta)$, including integral solutions and as a consequence  $\aff^\star(I,\Delta) \leq \min_{D}\phi^{\Delta}\leq u$.
In the second case, each coordinate of $(\lambda^\star,\mu^\star)$  is positive if and only if the same coordinate is positive at some point in the optimal polytope of the BLP. Therefore, the feasibility lattice of $\aff^\star(I,\Delta)$ includes every possible assignment which is in the support of some optimal
solution to $\blp(I,\Delta)$, including integral solutions and as a consequence $\aff^\star(I,\Delta) \leq \min_{D}\phi^{\Delta}=u$. Thus, in both cases, $\blp(I,\Delta)\leq u$ and $\aff^\star(I,\Delta)\leq u$ and hence Algorithm~\ref{fig:alg} indeed outputs \textsc{yes}, as required.

In the other direction, we want to show (by contrapositive) that if Algorithm~\ref{fig:alg} outputs \textsc{yes} then $\inf_C\phi^{\g}\leq u$.
Thus, assume that $\blp(I,\Delta)\leq u$ and $\aff^\star(I,\Delta)\leq u$. 
Let $(\lambda^\star, \mu^\star)$ be as in Definition~\ref{def:refinement} and denote $\blp^\star(I,\Delta):=\sum_{j \in J} \sum_{ t \in D^{n_j}} \lambda^\star_{j}(t) f^{\Delta}_j(t)$; observe that $\blp^\star(I,\Delta)\leq u$ by the definition of $(\lambda^\star, \mu^\star)$.  
Let $(q^\star,r^\star)$ be a solution to $\aff^\star(I,\Delta)$ with objective value at most $u$. Let
  $\ell$ be a positive integer such that $\ell\cdot \lambda^\star$, and $\ell\cdot \mu^\star$ are both integral, and let $M$ be  the maximum of the absolute values of the coordinates of both $q^\star$ and $r^\star$.
	Let us set $L:=(M+1)\ell$. 
	From $\blp^\star(I,\Delta)\leq u$ and $\aff^\star(I,\Delta)\leq u$ it immediately follows
  that
  \begin{equation}\label{eq:halfsym12}\frac{2(M+1)\ell}{2(M+1)\ell+1}\blp^\star(I,\Delta)+\frac{1}{2(M+1)\ell+1}\aff^\star(I,\Delta)\leq u.\end{equation}
	We claim that 
  \begin{equation}\label{eq:halfsym22}
    \min_{\multiset{D}{L+1}\times\multiset{D}{L}} \phi^{\mathcal
    B^{2L+1}_{B_1,B_2}(\Delta)}\ \leq\ 
  \frac{2(M+1)\ell}{2(M+1)\ell+1}\blp^\star(I,\Delta)+\frac{1}{2(M+1)\ell+1}\aff^\star(I,\Delta)
  \end{equation}
	for all the partitions $B_1\cup B_2=\left[2L+1\right]$ such that $\vert B_1\rvert=L+1$ and $\vert B_2\rvert=L$.
	To prove the claim, let us  define, for every $i \in \{1,\ldots, n\}$ and for every $a \in D$, the following nonnegative integers  \begin{align*}
		& W_{x_i,B_1}(a):=(M+1)\ell \mu^\star_{x_i}(a)+r^\star_{x_i}(a),\\
		& W_{x_i,B_2}(a):=(M+1)\ell \mu^\star_{x_i}(a).
	\end{align*}
	(To check  that $ W_{x_i,B_1}(a)$ and $W_{x_i,B_2}(a)$ are nonnegative it is
  enough to observe that if $\mu^\star_{x_i}(a)$ is $0$ then, by Definition~\ref{def:refinement},
  $r^\star_{x_i}(a)$ is also $0$, otherwise $\mu^\star_{x_i}(a)$ is at least
  $\frac{1}{\ell}$, and  the positivity of $ W_{x_i,B_1}(a)$, $W_{x_i,B_2}(a)$ follows by the choice of $M$.)
	Observe that for every $i \in \{1,\ldots, n\}$ we have that
	\begin{align*}
		&\sum_{a \in D} W_{x_i,B_1}(a)=(M+1)\ell \sum_{a \in D}\mu^\star_{x_i}(a)+\sum_{a \in D}r^\star_{x_i}(a)=(M+1)\ell+1=L+1,\\
		& \sum_{a \in D} W_{x_i,B_2}(a)=(M+1)\ell \sum_{a \in D}\mu^\star_{x_i}(a)=(M+1)\ell=L.
	\end{align*}
	Let  $\nu\colon V 
	\to \multiset{D}{L+1}\times\multiset{D}{L}$ be the map defined,  for every $x_i \in V$,  by
	$\nu (x_i)=(\alpha_i, \beta_i),$ where $\alpha_i$ is the multiset of $ \multiset{D}{L+1}$ that contains $ W_{x_i,B_1}(a)$ many occurrences of $a$, for every $a \in D$, and $\beta_i$ is the multiset of $ \multiset{D}{L}$ that contains $ W_{x_i,B_2}(a)$ many occurrences of $a$, for every $a \in D$.

	Let $f_j$ be a $k$-ary function symbol appearing as a term of the objective function $\phi$.
	Let us define,  for every $t \in D^{k}$, the following nonnegative integers  \begin{align*}
		& P_{j,B_1}(t):=(M+1)\ell \lambda^\star_{j}(t)+q^\star_{j}(t),\\
		& P_{j,B_2}(t):=(M+1)\ell \lambda^\star_{j}(t).
	\end{align*}
	Observe that 
	\begin{align*}
		&\sum_{t \in D^{k}} P_{j,B_1}(t)=(M+1)\ell \sum_{t \in D^{k}}\lambda^\star_{j}(t)+\sum_{t \in D^{k}}q^\star_{j}(t)=(M+1)\ell+1=L+1,\\
		& \sum_{t \in D^{k}} P_{j,B_2}(t)=(M+1)\ell \sum_{t \in D^{k}}\lambda^\star_{j}(t)=(M+1)\ell=L.
	\end{align*}
	We write now \begin{align*}
		& \sum_{t \in D^{k}}P_{j,B_1}(t)f^{\Delta}_j(t)=\sum_{h=1}^{(M+1)\ell+1}f^{\Delta}_j(\zeta_1^h,\ldots, \zeta_k^h)\\
	\end{align*}
	where $\zeta^1,\ldots, \zeta^{(M+1)\ell+1}$ are defined to be $(M+1)\ell+1$ elements of $D^k$ such that $P_{j,B_1}(t)$ many of them are equal to $t$, for every $t \in D^k$;
	and 
	\begin{align*}
		& \sum_{t \in D^{k}}P_{j,B_2}(t)f^{\Delta}_j(t)=\sum_{h=1}^{(M+1)\ell}f^{\Delta}_j(\xi_1^h,\ldots, \xi_k^h)\\
	\end{align*}
	where $\xi^1,\ldots, \xi^{(M+1)\ell}$ are defined to be $(M+1)\ell$ elements of $D^k$ such that $P_{j,B_2}(t)$ many of them are equal to $t$, for every $t \in D^k$. 
	
	We obtain
  \begin{align*}&\frac{2(M+1)\ell}{2(M+1)\ell+1}\lambda^\star_j(t)f_j^{\Delta}(t)+\frac{1}{2(M+1)\ell+1}q
    ^\star_j(t)f_j^{\Delta}(t)\\=&\frac{1}{2(M+1)\ell+1}\left(\sum_{t \in D^{k}}P_{j,B_1}(t)f^{\Delta}_j(t)+\sum_{t \in D^{k}}P_{j,B_2}(t)f^{\Delta}_j(t)\right)\\
		=&\frac{1}{2L+1}\left(\sum_{h=1}^{L+1}f^{\Delta}_j(\zeta_1^h,\ldots, \zeta_k^h)+\sum_{h=1}^{L}f^{\Delta}_j(\xi_1^h,\ldots, \xi_k^h)\right)\\
		\geq&\frac{1}{2L+1}\min_{t^1,\ldots,t^k \in D^m: \{t^{\ell}\}_{B_1}=\zeta_{\ell}, \{t^{\ell}\}_{B_2}=\xi_{\ell}}\sum_{i=1}^{2L+1}f^{\Delta}(t^1_i,\ldots,t^k_i)\\=&f_j^{\mathcal B^{2L+1}_{B_1,B_2}(\Delta)}((\zeta_1,\xi_1),\ldots,(\zeta_k,\xi_k))=f_j^{\mathcal B^{2L+1}_{B_1,B_2}(\Delta)}(\nu(x^j_1),\ldots, \nu(x_k^j)),
	\end{align*}
	where the last equality follows because, for every $a \in D$, the number of $a$'s in $\zeta_h$ is 
	\begin{align*} &\sum_{ t \in D^{k}: t_h=a}P_{j,B_1}(t)=(M+1)\ell \sum_{ t \in
  D^{k}: t_h=a}\lambda^\star_j(t)+\sum_{ t \in D^{k}:
  t_h=a}q^\star_j(t)\\=&(M+1)\ell\mu^\star_{x_h^j}(a)+r^\star_{x_h^j}(a)=W_{x_h^j,B_1}(a),\end{align*}
	and, 
	for every $a \in D$, the number of $a$'s in $\xi_h$ is 
    \[\sum_{ t \in D^{k}: t_h=a}P_{j,B_2}(t)=(M+1)\ell \sum_{ t \in
  D^{k}: t_h=a}\lambda^\star_j(t)
  =(M+1)\ell\mu^\star_{x_h^j}(a)=W_{x_h^j,B_2}(a).\]
  This proves the claim.
	
	From Inequalities (\ref{eq:halfsym12}) and (\ref{eq:halfsym22}) it follows
  that for all partitions $B_1\cup B_2=\left[2L+1\right]$ such that $\vert
  B_1\rvert=L+1$ and $\vert B_2\rvert=L$ it holds
  \[\min_{\multiset{D}{L+1}\times\multiset{D}{L}} \phi^{\mathcal
  B^{2L+1}_{B_1,B_2}(\Delta)}\leq u.\]  Moreover, since there exists a
  block-symmetric promise fractional polymorphisms of $(\Delta,\g)$ of arity
  $2L+1$ having two symmetric blocks $B_1$ and $B_2$ with respective size $L+1$
  and $L$, Lemma~\ref{lemma:fhomhalfsym2} implies the existence of a fractional
  homomorphism from $\mathcal B^{2L+1}_{B_1,B_2}(\Delta)$ to $\g$. From Proposition \ref{prop:frachom} it follows that
  \[\inf_{C}\phi^{\g}\leq \min_{\multiset{D}{L+1}\times\multiset{D}{L}} \phi^{\mathcal B^{2L+1}_{B_1,B_2}(\Delta)}\leq u\]
	and this concludes the proof.
\end{proof}

We conclude the section with a couple of remarks showing that the number of
symmetric blocks and their size do not play a crucial role in the proofs of
Theorem~\ref{thm:main} and Lemma~\ref{lemma:fhomhalfsym2}. The notion of bimultiset-structure can be straightforwardly generalised to the notion of $k$\emph{-multiset-structure with blocks of size} $b_1,\ldots, b_k$.   

\begin{remark}
The proof of Lemma~\ref{lemma:fhomhalfsym2} can be easily adapted to the case in which the promise template $(\Delta, \g)$ has a block-symmetric promise fractional polymorphism $\omega$ of arity $L \in \mathbb N$ with an arbitrary number $k$ of symmetric blocks of size $b_1,\ldots,b_k$ (given that the blocks are the same for each map in $\supp(\omega)$). In this case, the $k$-multiset structure $\mathcal B_{B_1,\ldots,B_k}(\Delta)$ is fractionally homomorphic to $\g$.
The converse (corresponding to Remark~\ref{rem:converse}) is also true.
\end{remark}

\begin{remark}
Let  $(\Delta,\g)$ be a promise valued template such that  $\Delta$ has a \emph{finite domain}. Assume that $(\Delta,\g)$ has block-symmetric promise fractional polymorphisms with arbitrarily many blocks of arbitrarily large size. Then Algorithm~\ref{fig:alg}  correctly solves $\pvcsp(\Delta,\g)$. 
This statement can be proved by a slight modification of the proof of Theorem~\ref{thm:main} employing a block-symmetric promise fractional polymorphism whose symmetric blocks $B_1,\ldots,B_k$ have each size at least $M\ell^2$ and where the coefficients $W_{x_i,B_b(a)}$, and $P_{j,B_b}(t)$ are defined as in the proof of \cite[Theorem~4.1]{BrakensiekGWZ20}.
\end{remark}

Finally, we point out that for finite-domain PCSPs, the condition of having a block-symmetric promise polymorphism of arity $2L+1$ is equivalent to the one of having a block-symmetric promise polymorphism of arity $2L+1$ with two symmetric blocks of arity $L+1$ and $L$, for every $L \in \mathbb N$~\cite{BrakensiekGWZ20}.

\section{The Combined Relaxation with Sampling for PVCSPs}

We use the notion of a sampling algorithm for a valued structure from~\cite{PLVCSPsolvbyLP}.

\begin{definition}\label{def:valued-sampling}
	Let $\g$ be a valued structure with domain $C$ and finite signature $\tau$. A \emph{sampling algorithm} for $\g$ takes as input a positive integer $d$  and computes a  finite-domain valued  structure  $\Delta$  fractionally homomorphic to $\g$ 
	such that, for every finite sum $\phi$ of  $\tau$-terms having at most $d$ distinct variables, $V=\{x_1,\ldots,x_d\}$, and every $u \in \Q$,  there exists a solution $s\colon V \to C$ with $\phi^{\g} (s(x_1),\ldots, s(x_d))\leq u$ if and only if there exists a solution $h'\colon V \to D$ with $\phi^{\Delta}(h'(x_1),\ldots, h'(x_d))\leq u$.
	A sampling algorithm is called \emph{efficient} if its running time is bounded by a polynomial in $d$.
	The finite-domain valued  structure computed by a sampling algorithm is called a \emph{sample}.
\end{definition}

\begin{example}\label{ex:samplingplh}
	A valued structure $\g$ with domain $\Q$ and signature $\tau$ is called \emph{piecewise linear homogeneous (PLH)} if, for every $\gamma \in \tau$, the cost function $\gamma^{\g}$ is first-order definable over the structure $\mathfrak {L}=(\Q;\leq,1,\{c\cdot\}_{c \in \Q})$ where 
	\begin{itemize}
    \item $<$ is a relation symbol (i.e., a $\{0,\infty\}$-valued function symbol) of arity $2$
		and $<^{\mathfrak L}$ is the strict linear order of $\Q$,
		\item $1$ is a constant symbol and $1^{\mathfrak L} := 1 \in \Q$, and 
		\item $c \cdot$ is a unary function symbol for every $c \in {\Q}$ such that $(c \cdot)^{\mathfrak L}$ is the function $x \mapsto cx$, i.e., the multiplication by $c$.
	\end{itemize} 
	If $\g$ is a PLH valued structure with a finite signature, then it admits an efficient sampling algorithm~\cite{PLVCSPsolvbyLP}.
\end{example}

Let $\g$ be a valued structure with a finite signature that admits an efficient sampling algorithm. Observe that for every finite-domain valued structure $\Delta_d$ computed by such an efficient sampling algorithm, the pair $(\Delta_d,\g)$ is a promise valued template. The following lemma  follows from the definition of a sampling algorithm for a valued structure.

\begin{lemma}\label{lemma:samplingvspromise}
	Let $(\g_1,\g_2)$ be a promise valued template with a finite signature.
  Assume that $\Gamma_1$ admits an efficient sampling algorithm.
  If  $\pvcsp(\Delta_d,\g_2)$ is polynomial-time solvable for every finite-domain
  valued structure $\Delta_d$ computed by  an efficient sampling algorithm for
  $\g_1$, then $\pvcsp(\g_1,\g_2)$ is polynomial-time solvable.
\end{lemma}

\begin{proof}  

It can be verified that $(\Delta_d,\g_2)$ is a valid promise valued template, for every $d \in \mathbb N$; that is, the existence of a fractional homomorphism from $\Delta_d$ to $\g_1$ and a fractional homomorphism from $\g_1$ to $\g_2$ implies the existence of a fractional homomorphism from $\Delta_d$ to $\g_2$, which is obtained by a composition of the two fractional homomorphisms in the same fashion as in the proof of Theorem~\ref{thm:main2}.

Let $I:=(V,\phi,u)$ be an instance of	$\pvcsp(\g_1,\g_2)$ and let $d:=\lvert V \rvert$. Let $C_1,C_2$ and $D_d$ be the domains of $\g_1, \g_2$ and $\Delta_d$, respectively. Assume that there is a polynomial-time algorithm solving $\pvcsp(\Delta_d,\g_2)$. If the output of such an algorithm is $\textsc{no}$ then clearly the answer to $\pvcsp(\g_1,\g_2)$ is $\textsc{no}$.
Otherwise, if the output of such an algorithm is $\textsc{yes}$, then by the definition of sampling algorithm the answer to  $\pvcsp(\g_1,\g_2)$ is $\textsc{yes}$.
\end{proof}

\begin{theorem}\label{thm:main2}
	Let $(\g_1,\g_2)$ be a promise valued template with a finite signature. 
  Assume that $\g_1$ admits an efficient sampling algorithm. 
  Moreover, assume that $(\g_1,\g_2)$ has a block-symmetric promise
  fractional polymorphism of arity $2L+1$ with two symmetric blocks of size
  $L+1$ and $L$, respectively, for all $L \in \N$. Then $\pvcsp(\g_1,\g_2)$ is
  polynomial-time solvable. 
\end{theorem}

Note that in Theorem~\ref{thm:main2} (and in Lemma~\ref{lemma:samplingvspromise}) {\it both} the respective domains of the valued structures $\g_1$ and $\g_2$ have arbitrary (countable) cardinality, that is, each of $\g_1$ and $\g_2$ can have a finite or an infinite domain.

\begin{proof}
For every positive integer $d$ let $\Delta_d$ be the finite-domain valued structure $\Delta_d$ computed by  an efficient sampling algorithm for $\g_1$ on input $d$. 
  By Lemma~\ref{lemma:samplingvspromise}, to prove that $\pvcsp(\g_1,\g_2)$ is
  polynomial-time solvable it is enough to prove that $\pvcsp(\Delta_d,\g_2)$ is
  polynomial-time solvable for every positive integer $d$.
We claim that for every $d, L \in \N$  there exists a block-symmetric promise
  fractional polymorphism of $(\Delta_d,\g_2)$ with arity $2L+1$ having two
  symmetric blocks of size $L+1$ and $L$, respectively. By Theorem~\ref{thm:main}, 
  proving our claim will imply that
  $\pvcsp(\Delta_d, \g_2)$ is polynomial-time solvable and therefore, by
  Lemma~\ref{lemma:samplingvspromise}, that $\pvcsp(\g_1,\g_2)$ is polynomial-time
  solvable. 

We now prove the claim. Let $C_1,C_2$ and $D_d$ be the domains of $\g_1, \g_2$ and $\Delta_d$, respectively, and let $\tau$ be the signature of $\g_1$ and $\g_2$ (by the definition of  sampling algorithm it is also the signature of $\Delta_d$).
By the definition of sampling algorithm, there exits a fractional homomorphism $\chi$ from $\Delta_d$ to $\g_1$ and, by assumption, there exists a block-symmetric promise fractional polymorphism $\omega:=(\omega_I,\omega_O)$ of $(\g_1,\g_2)$ with two symmetric blocks $B_1$ and $B_2$ having size $L+1$ and $L$, respectively. For every $h \in \supp(\chi)\subseteq C_1^{D_d}$ and every $g \in \supp(\omega_O)\subseteq C_2^{C_1^{2L+1}}$ we define the map $g \circ h \colon D_d^{2L+1}\to C_2$ by setting, for every $a^1, \ldots,a^m \in D_d^{2L+1}$, \[g \circ h (a^1,\ldots,a^{2L+1})=g\left(h(a^1), \ldots, h(a^{2L+1})\right).\]
Observe that $g \circ h$ is a block-symmetric map with symmetric block $B_1$ and $B_2$, because so is $g$. We define the discrete probability measure $\mu_O$ on $C_2^{D_d^{2L+1}}$ as follows
\[\mu_O(Y)= \sum_{h \in \supp(\chi)} \sum_{\substack{g \in \supp(\omega_O):\\  g \circ h\in Y}} \chi(h)\omega_O(g),\qquad\text{ for every } Y \subseteq D_d^{2L+1}.\]
Since $\sum_{h \in \supp(\chi)}\chi(h)$ and $\sum_{g \in \supp(\omega_O)} \omega_O(g)$ are absolutely convergent series, it follows  by Martens' Theorem (see, e.g.,~\cite[Theorem 6.57]{laczkovich2017real}) that $\mu_O$ is well defined.
Observe that $\mu_O$ satisfies the countable additivity property, since $\omega_O$ and $\chi$ do. Furthermore, $\supp(\mu_O)=\{(g\circ h)\mid h \in \supp(\chi),  g \in \supp(\omega_O)\}$ is countable  and, again by Martens' Theorem,  it holds
\begin{align*}\sum_{g' \in \supp(\mu_O)}\mu_O(g')= \sum_{h \in \supp(\chi)}\chi(h)\sum_{g \in \supp(\omega_O)} \omega_O(g)=1.\end{align*}
Let $\mu_I$ be the discrete probability measure in $\mathcal J_D^{(2L+1)}$ such
  that $\mu_I(e_i^{(2L+1)})=\frac{1}{2L+1}$. We now show that
  $\mu:=(\mu_I,\mu_O)$ is a promise fractional polymorphism of $(\Delta_d, \g_1)$
  with arity $2L+1$. By definition, $\mu$ is block-symmetric with two symmetric blocks of respective size $L+1$ and $L$.

For every $f \in \tau$ and every $a^1, \ldots, a^{2L+1} \in D_d^{\ar(f)}$ it holds that 
\begin{align}\nonumber
&\sum_{ g' \in \supp(\mu_O)}\mu_O(g')f^{\g_2}\left(g(a^1,\ldots, a^{2L+1})\right)\\\nonumber=&\sum_{h \in \supp(\chi)}\chi(h)\sum_{g \in \supp(\omega_O)}\omega_O(g)f^{\g_2}\left(g(h(a^1), \ldots,h(a^{2L+1})\right)\\
\label{eq:row3}\leq &\sum_{h \in \supp(\chi)}\chi(h) \frac{1}{2L+1}\sum_{i=1}^{2L+1}f^{\g_1}\left(g(h(a^i)) \right)\\
\nonumber = & \frac{1}{2L+1} \sum_{i=1}^{2L+1} \sum_{h \in \supp(\chi)}\chi(h)f^{\g_1}\left(h(a^i)\right)\\
\label{eq:row4}\leq & \frac{1}{2L+1}\sum_{i=1}^{2L+1}f^{\Delta_d}(a^i),
\end{align}
where Inequality~(\ref{eq:row3}) holds  because $\omega$ is a promise fractional
  polymorphism of $(\g_1,\g_2)$ and Inequality~(\ref{eq:row4}) holds  because
  $\chi$ is a fractional homomorphism from $\Delta_d$ to $\g_1$.
We obtained that $\mu$ is a promise fractional polymorphism of $(\Delta_d, \g_2)$ and this concludes the proof.
\end{proof}

In the particular case in which $\g_1=\g_2$, from Theorem~\ref{thm:main2} we obtain the following result for infinite-domain (non-promise) VCSPs.

\begin{corollary}\label{cor:main2}
	Let $\g$ be a valued structure with a finite signature that admits an
	efficient sampling algorithm. Assume that $\g$ has a block-symmetric
	fractional polymorphism of arity $2L+1$ with two symmetric blocks of size
	$L+1$ and $L$, respectively, for all $L \in \N$. Then $\vcsp(\g)$ is
	polynomial-time solvable. 
\end{corollary}

Corollary~\ref{cor:main2} and the existence of an efficient sampling algorithm for PLH valued structures (see Example~\ref{ex:samplingplh}) imply the following tractability result for $(\wma, \avg)$-convex PLH valued structures.

\begin{corollary}Let $\wma$ be the $2$-period centred weighted  moving average.
  Let $\g$ be a PLH valued structure with a finite signature such that
  every cost function from $\g$ is a $(\wma, \avg)$-convex cost function (see
  Example \ref{ex:movingaverage2}). Then $\vcsp(\g)$ is polynomial-time solvable. \qedhere
\end{corollary}

\bibliographystyle{plainurl}
\bibliography{vz21}

\appendix

\section{Omitted Proofs from Section~\ref{sec:prelims}}
\label{app:prelims}

Let $\mu$ be a discrete probability measure as defined in
Section~\ref{sec:prelims}. We note that countable additivity implies that
$\mu(\emptyset)=0$ and that $\mu(X_1)\leq \mu(X_2)$ for every $X_1\subseteq
X_2$, i.e., $\mu$ is monotone. We also remark that while the set $X$ might be
countable or uncountable, the terminology \emph{discrete probability measure}
refers to the fact that the $\sigma$-algebra is $\mathcal P(X)$, i.e., the
discrete topology on $X$ (see~\cite[Examples 2.8 and 2.9]{Loasve1977}).

\begin{proposition*}[Proposition~\ref{prop:prob} restated]
Let $\mu$ be a discrete probability measure on a  set $X$. Then  $\supp(\mu)$  is a countable  set. Furthermore, if $X$ is countable then $\sum_{x \in \supp(\mu)}\mu(x)=1$; that is, $\supp(\mu)$ is non-empty.
\end{proposition*}

\begin{proof}
	Let us define, for all $n \in \N$, the subset $A_n:=\left\{x \in X \mid \mu(\{x\})\geq \frac{1}{n}\right\}$. We can write the support of $\mu$ as $\supp(\mu)=\bigcup_{n \in \N}A_n$. Observe that, for every $n \in \N$, the set $A_n$ is finite. If this was not the case, then there would exist a sequence $(x_i)_{i \in \N}\in (A_n)^{\N}$ of pairwise distinct elements and, by the countable additivity property, we would obtain that $\mu (\{ x_i \mid i \in \N \} )$ $=\sum_{i \in \N}\mu(\{x_i\})\geq \sum_{i \in \N}\frac{1}{n}=+\infty$, which contradicts the assumption that
	$\mu (\{x_i \mid i \in \N\})\leq \mu(X)=1$. Therefore, $\supp(\mu)$ is countable, because it is the countable union of finite sets.
\end{proof}

\begin{proposition*}[Proposition~\ref{prop:frachom} restated]
	Let $\g$ and $\Delta$ be valued structures over the same signature $\tau$ with domains $C$ and $D$, respectively. Assume $\Delta \to_f \g$. Let $V=\{v_1,\ldots,v_n\}$ be a set of variables and $\phi$ a sum of finitely many $\tau$-terms with variables from $V$. For every $u\in \Q$, if there exists an assignment $s \colon V \to D$ such that $\phi^{\Delta}(s(v_1, \ldots, s(v_n)) \leq u$, then there exists an assignment $s'\colon V \to C$ such that $\phi^{\g}(s'(v_1), \ldots, s'(v_n)) \leq u$. In particular, it holds that \[\inf_{C} \phi^{\g}\leq \inf_{D} \phi^{\Delta}.\] 
\end{proposition*}

\begin{proof}
	Let $\phi(v_1,\ldots,v_n):=\sum_{j \in J}\gamma_j(v^j)$, where $\gamma_j \in
  \tau$ and $v^j \in V^{ar(\gamma_j)}$; and let $u\in \Q$. Let $\chi$ be a
  fractional homomorphism from $\Delta$ to $\g$ and let $s \colon V \to D$ be an
  assignment with cost $\phi^{\Delta}(s(v_1), \ldots, s(v_n)) \leq u$. Then, by the definition of fractional homomorphism, \begin{align*}&\sum_{h \in \supp(\chi)}\chi(h)\sum_{j \in J}\gamma_j^{\g}(h(s(v^j)))=\sum_{j \in J}\sum_{h \in \supp(\chi)}\chi(h)\gamma_j^{\g}(h(s(v^j)))\leq\sum_{j \in J}\gamma_j^{\delta}(s(v^j))\leq u.\end{align*}
	Therefore, there exists at least one map $h\colon D\to C$ from $\supp(\chi)$, such that $\phi^{\g}(h\circ s(v_1), \ldots, h\circ s(v_n))\leq\phi^{\Delta}(s(v_1), \ldots, s(v_n))\leq u$, that is, $h\circ s\colon V\to C$ is an assignment with cost at most $u$.
\end{proof}

\begin{lemma*}[Lemma~\ref{lemma:weightsyminput} restated]
Let $(\Delta,\g)$ be a promise valued template and let $m \in \N$. If $\omega=(\omega_I, \omega_O)$ is an $m$-ary block symmetric promise fractional polymorphism of $(\Delta,\g)$, then also $\omega'=(\omega'_I, \omega_O)$, where $\omega'_I(e^{(m)}_i)=\frac{1}{m}$ for $1\leq i \leq m$, is an $m$-ary block-symmetric promise fractional polymorphism of $(\Delta,\g)$.
\end{lemma*}

\begin{proof}
	Let $C$ and $D$ be the domains of $\g$ and $\Delta$, respectively, and let $\tau$ be the common signature of the two valued structures. By the definition of block-symmetric promise fractional polymorphism, there exists a partition $B_1\cup \cdots \cup B_k$ of $\left [m\right]$ such that for every permutation $\pi \in S_m$ that preserves the blocks $B_1, \ldots,B_k $, every $f \in \tau$ and every $a^1,\ldots,a^m \in D^{ar(f)}$ it holds that
	\begin{align*}
	\sum_{g \in \supp(\omega_O)}\omega_O(g)f^{\g}(g(a^1,\ldots,a^m))
	\leq \sum_{i=1}^m\omega_I(e_i^{(m)})f^{\Delta}(a^{\pi(i)}).
	\end{align*}
	By summing the last inequality over all $\pi \in S_m$ that preserve all the blocks $B_1,\ldots,B_k$, i.e., over all $\pi=\pi^1\cdots\pi^k$ such that $\pi^j$ is a permutation of the elements in $B_j$ for $1 \leq i \leq k$, we obtain
	\begin{align*}
    \prod_{j =1}^{k}\lvert B_j \rvert ! \sum_{g \in \supp(\omega_O)}f^{\g}(g(a^1,\ldots,a_m)) 
  \;\leq\; 
    \sum_{j=1}^k \prod_{h \neq j} \lvert B_h \rvert ! \sum_{i \in B_j}\omega_I(e^{(m)}_i)\left((\lvert Bj\rvert-1)! \sum_{{\ell} \in B_j}f^{\Delta}(a^{\ell})\right)\\
	\;=\prod_{j =1}^{k}\lvert B_j \rvert ! \sum_{j =1}^k \frac{1}{\lvert B_j\rvert} \sum_{i \in B_j} \omega_I(e_i^{(m)}) \sum_{{\ell} \in B_j} f^{\Delta}(a^{\ell}).
	\end{align*}Since, by the definition of block-symmetric promise fractional polymorphism, $\sum_{i \in B_j} \omega_I(e_i^{(m)})=\frac{\lvert B_j\rvert}{m}$ for every $j \in \{1,\ldots, k\}$, we obtain
	\begin{align*}
	\; \sum_{g \in \supp(\omega_O)}f^{\g}(g(a^1,\ldots,a^m)) \leq 
	\frac{1}{m}\sum_{i=1}^mf^{\Delta}(a^{\ell}),
	\end{align*}
	that is, $\omega'$ is a promise fractional polymorphism of $(\Delta,\g)$.
\end{proof}

\section{The BLP Relaxation}
\label{app:blp}

A fractional polymorphism of a valued structure $\g$ is \emph{fully symmetric} if every operation in its support is a fully symmetric map. 
Given a promise valued template $(\Delta,\g)$, a promise fractional polymorphism $\omega=(\omega_I,\omega_O)$ of $(\Delta,\g)$ is \emph{fully symmetric} if every map in $\supp(\omega_O)$ is fully symmetric.

The following lemma is a corollary of Lemma~\ref{lemma:weightsyminput}.

\begin{lemma}\label{lemma:weightsyminput1}
Let $(\Delta,\g)$ be a promise valued template and let $m \in \N$. If $\omega=(\omega_I, \omega_O)$ is an $m$-ary fully symmetric promise fractional polymorphism of $(\Delta,\g)$, then also $\omega'=(\omega'_I, \omega_O)$, where $\omega'_I(e^{(m)}_i)=\frac{1}{m}$ for $1\leq i \leq m$, is an $m$-ary fully symmetric promise fractional polymorphism of $(\Delta,\g)$.
\end{lemma}

In view of Lemma~\ref{lemma:weightsyminput1}, 
we will assume without loss of generality that any
$m$-ary  fully symmetric promise fractional polymorphism $\omega=(\omega_I,
\omega_O)$ is such that $\omega_I$ assign $\frac{1}{m}$ to each $m$-ary
projection on the domain of $\Delta$ and we will identify $\omega$ with $\omega_O$.

Recall the BLP relaxation from Section~\ref{subsec:blp}. 
Let $(\Delta,\g)$ be a promise valued template such that the domain of $\Delta$
is a \emph{finite} set. We may solve $\pvcsp(\Delta,\g)$ by using the following
algorithm that computes a BLP relaxation of $\Delta$. 

\vskip \baselineskip
\begin{algorithm}[H] 
	\SetAlgoNoLine
	\KwIn{\\\qquad $I:=(V,\phi,u)$, a valid instance of $\pvcsp(\Delta,\g)$} 
  \KwOut{\\\qquad\textsc{yes} if there exists an assignment $s \colon V\to \dom(\Delta)$ such that $\phi^{\Delta}(s(x_1),\ldots,s(x_{\lvert V\rvert }))\leq u$\\\qquad \textsc{no} if there is no assignment $s \colon V\to \dom(\g)$ such that $\phi^{\Gamma}(s(x_1),\ldots,s(x_{\lvert V\rvert }))\leq u$}
  \medskip
	$\blp(I, \Delta)$\;
	\eIf{$\blp(I,\Delta)\nleq u$}{output \textsc{no}\;
  }{output \textsc{yes}\;}
	
	\caption{BLP  Relaxation Algorithm for $\pvcsp(\Delta,\g)$}
	\label{figure:algblp}
\end{algorithm}

\vskip \baselineskip
Algorithm \ref{figure:algblp} runs in polynomial time in $I$, and that if it outputs \textsc{no}, then indeed the answer to $\pvcsp(\Delta,\g)$ is \textsc{no}, without further assumptions.

We now present a sufficient condition under which Algorithm \ref{figure:algblp} correctly solves $\pvcsp(\Delta,\g)$. 

\begin{theorem}\label{thm:fullysymfpol}
	Let  $(\Delta,\g)$ be a promise valued template such that  $\Delta$ has a
  \emph{finite domain}. Assume that for all $m \in \N$ there exists a  fully symmetric promise fractional polymorphisms of $(\Delta,\g)$ with arity $m$. Then Algorithm \ref{figure:algblp}  correctly solves $\pvcsp(\Delta,\g)$ (in polynomial time). 
\end{theorem}

To prove Theorem~\ref{thm:fullysymfpol} we need to use a preliminary lemma and the notion of a \emph{multiset-structures}.

Let $\Delta$ be a valued $\tau$-structure with domain $D$ and  let $m \in \N$.
The \emph{multiset-structure} $\mathcal P^{m}(\Delta)$   is the valued structure
with domain $\multiset{D}{m}$ i.e., the set whose elements $\alpha$ are
multisets of elements from $D$ of size $m$. For every $k$-ary function symbol
$f\in \tau$, and $\alpha_1,\ldots,\alpha_k \in \multiset{D}{m}$ the function $f^{\mathcal P^{m}(\Delta)}$ is defined as follows
\begin{align*}f^{\mathcal P^{m}(\Delta)}(\alpha_1,\ldots, \alpha_k):=\frac{1}{m}\min_{\substack{t^1,\ldots,t^k \in D^{m}:\\ \{t^{\ell}\}=\alpha_{\ell}}}\sum_{i=1}^{m}f^{\Delta}(t^1_i,\ldots,t^k_i),\end{align*}
where $\{t^{\ell}\}$ is the multiset of the coordinates of $t^{\ell}$, i.e., the multiset $\{t^{\ell}_i \mid 1\leq i \leq m\}$.

\begin{lemma}\label{lemma:fhom}
	Let  $(\Delta,\g)$ be a promise valued template such that  $\Delta$ has a \emph{finite domain}. Let $m\in \N$ and let us assume that $(\Delta,\g)$ has a  fully symmetric promise fractional polymorphism of arity $m$. Then $\mathcal P^{m}(\Delta)$ is fractionally homomorphic to $\g$. 
\end{lemma}

Lemma~\ref{lemma:fhom} is a generalisation of~\cite[Lemma~2.2]{ThapperZivny2012} to promise valued templates ($\Delta, \g)$ where $\Delta$ is a finite-domain valued structure. Furthermore, it already appeared in a similar form in~\cite[Lemma~6.9]{PLVCSPsolvbyLP}.\footnote{In~\cite{PLVCSPsolvbyLP} the assumption of having an $m$-ary fully-symmetric promise fractional polymorphism of $(\Delta,\g)$ is replaced by the hypothesis of having an $m$-ary fully-symmetric fractional polymorphism of $\g$. Another difference is that in~\cite{PLVCSPsolvbyLP} the notion of fractional polymorphism employed allows only finitely supported probability measures.} 

\begin{proof}[Proof of Lemma~\ref{lemma:fhom}]
  Let $C$ be the (possibly infinite) domain of $\g$, let $D$ be the finite domain of $\Delta$, and let $\tau$ be the common signature of $\g$ and $\Delta$.
	Let  $\omega$ be the $m$-ary  fully symmetric promise fractional polymorphism
  of $(\Delta,\g)$. For every $g \in \supp(\omega)\subseteq C^{D^m}$  we define
  $\tilde g\colon\multiset{D}{m}\to C$ by setting, for every  $\alpha=\{\xi^1,
  \ldots, \xi^m\}\in \multiset{D}{m}$, 
	\[\tilde{g}(\alpha)=g(\xi^1,\xi^2,\ldots,\xi^{m}).\] 
	Observe that $\tilde{g}$ is well defined as $g$ is  fully symmetric.
  We define the discrete probability measure $\chi$ on $C^{\multiset{D}{m}}$ as
  follows \[\chi(Y)= \sum_{\substack{g \in \supp(\omega):\\  \tilde{g} \in Y}}
  \omega(g),\qquad\text{ for every } Y \subseteq \mmultiset{D}{m}.\]
	
	Observe that $\chi$ satisfies the countable additivity property, since $\omega$ does. Furthermore, $\supp(\chi)=\{h \in C^{D^m} \mid  h= \tilde{g} \text{ for some } g \in \supp(\omega)\}=\{\tilde{g} \in C^{D^m} \mid  g \in \supp(\omega)\}$ is countable  and  it holds that
	\begin{align*}\sum_{h \in \supp(\chi)}\chi(h)= \sum_{g \in \supp(\omega)} \omega(g)=1.\end{align*}
	We claim that $\chi$ is a fractional homomorphism from ${\mathcal P^{m}(\Delta})$ to $\g$.
	Indeed, for every $f \in \tau$ and every  tuple $(\alpha_1,\ldots, \alpha_k)\in\multiset{D}{m}^k$ with $\alpha_i:=\{\xi^1_i,\ldots, \xi^m_i\}$ and $k:=\ar(f)$, it holds that 
	{ \begin{align}\nonumber
		& \sum_{h \in C^{\multiset{D}{m}}}\chi(h)f^{\g}(h(\alpha_1)\ldots,h(\alpha_k))\\
   \nonumber= & \sum_{g \in \supp(\omega)}\omega(g) f^{\g}(g(\xi_1^1,\ldots,\xi_1^{m}),\ldots,g(\xi_k^1,\ldots,\xi_k^{m}))\\\label{eq:frhom1}
		=& \sum_{g \in \supp(\omega)}\omega(g)f^{\g}(g(\xi_1^{\pi_1(1)},\ldots,\xi_1^{\pi_1({m})}),\ldots,g(\xi_k^{\pi_k(1)},\ldots,\xi_k^{\pi_k({m})}))\\\label{eq:frhom2}
		\leq &\; \frac{1}{m}\sum_{i=1}^{m} f^{\Delta} (\xi^{\pi_1(i)}_1, \ldots, \xi^{\pi_k(i)}_k)
		\end{align}}
	for every $\pi_1, \ldots,\pi_k \in S_{m}$. 
	Equality~(\ref{eq:frhom1}) holds because the maps $g \in \supp(\omega)$ are  fully symmetric and Inequality~(\ref{eq:frhom2}) holds because $\omega$ is a promise fractional polymorphism of $(\Delta,\g)$.
	Then, in particular, we obtain  \begin{align*}
	&\sum_{h \in C^{\multiset{D}{m}}}\chi(h)f^{\g}(h(\alpha_1,\ldots,\alpha_k))\leq   \frac{1}{m}\min_{\substack{t^1,\ldots,t^k \in D^{m}:\\ \{t^{\ell}\}=\alpha_{\ell}}}\sum_{i=1}^{m}f^{\Delta} (t_{i}^1, \ldots t_{i}^k)\\=& f^{\mathcal{P}^{m}(\Delta)}(\alpha_1,\ldots,\alpha_k).\qedhere
\end{align*}\end{proof}	

With Lemma~\ref{lemma:fhom}, the proof of Theorem~\ref{thm:fullysymfpol} follows
the same argument as in the finite-domain
non-promise case~\cite[Theorem~3.2]{ThapperZivny2012}. We include the proof here for
completeness. 

\begin{proof}[Proof of Theorem~\ref{thm:fullysymfpol}]
	Let $C$ be the (possibly infinite) domain of $\g$ and let $D$ be the finite domain of $\Delta$. Let $\tau$ be the common signature of $\Delta$ and $\g$.	
	Let $I$ be an instance of $\pvcsp(\Delta,\g)$ with variables $V=\{x_1, \ldots,x_n\}$, objective function $\phi(x_1,\ldots,x_n)=\sum_{j \in J}\gamma_j(x^j)$ where $J$ is a finite set of indices, $\gamma_j \in \g$, and $x^j \in V^{\ar(j)}$, and threshold $u$. 
 Note that if $\blp(I,\Delta) \nleq u$ (this also includes the case $\blp(I,\Delta)=+\infty$, i.e., the case that $I$ is not feasible) then $\min_{D}\phi^{\Delta}\nleq u$.  We may therefore safely output \textsc{no}.
	Otherwise, $\blp(I,\Delta)\leq u$ implies that 
  $\inf_{C}\phi^{\g}\leq \blp(I,\Delta)$. The proof of this last
  statement is contained in the first part of the proof
  of~\cite[Theorem~3.2]{ThapperZivny2012}; we include it here for completeness. 
	Let $(\lambda^\star, \mu^\star)$ be an optimal solution to $\blp(I, \Delta)$
  and let $M$ be a positive integer such that $M\cdot \lambda^\star$, and
  $M\cdot \mu^\star$ are both integral. Let $\nu\colon V \to \multiset{D}{M}$ be defined by mapping the variable $x_i$ to the multiset in which the elements are distributed accordingly to $\mu_{x_i}^\star$, i.e., for every $a \in D$ the number of occurrences of $a$ in $\nu(x_i)$ is equal to $M\mu_{x_i}^\star(a)$. Let $f_j$ be a $k$-ary function symbol in $\tau$ that occurs in a term $f_j(x^j)$ of the objective function $\phi$. 
	Now we write \[M\cdot \sum_{t \in D^k}\lambda^\star_j(t)f^{\Delta}_j(t)=f^{\Delta}_j(\alpha^1)+\cdots+f^{\Delta}_j(\alpha^M)\text{,}\] where the $\alpha^i \in D^k$ are such that $\lambda^\star_j(t)$-fractions are equal to $t$. Let us define $\alpha_{\ell}':=(\alpha^1_i,\ldots,\alpha^M_i)$ for $1\leq i \leq k$. We get 
	\begin{align*}&\sum_{ t \in D^{k}}\lambda^\star_j(t)f_j^{\Delta}(t)=\frac{1}{M}\sum_{i=1}^{M}f_j^{\Delta}(\alpha^i)=\frac{1}{M}\sum_{i=1}^{M}f_j^{\Delta}(\alpha_1^i, \ldots, \alpha_k^i)\\ \geq&  \frac{1}{M}\min_{\substack{t^1,\ldots,t^k \in D^M:\\ \{t^{\ell}\}=\{\alpha_{\ell}'\}}}\sum_{i=1}^{M}f_j^{\Delta}(t^1_i,\ldots,t^k_i)=f_j^{\mathcal P^M(\Delta)}(\alpha'_1,\ldots,\alpha'_k)\\=&f_j^{\mathcal P^M(\Delta)}(\nu(x))\text{, }\end{align*}
	where the last equality follows as the number of $a$'s in $\alpha_i'$ is \[M \cdot \sum_{ t \in D^{k}: t^i=a}\lambda^\star_j(t)=M\cdot \mu^\star_{x_i}(a).\] Then \begin{align*}
	\blp(I, \Delta)=&\sum_{ j \in J}\sum_{t \in
  D^{\ar(f_j)}}\lambda^\star_j(t)f_j^{\Delta}(t) = \sum_{ j \in J}\left (\sum_{t
  \in D^{\ar(f_j)}}\lambda^\star_j(t)f_j^{\Delta}(t) \right)\\ \geq&   \sum_{ j
  \in J}\left (f_j^{\mathcal P^M(\Delta)}(\nu(x))\right) \geq \min_{
  \multiset{D}{m}}\phi^{\mathcal{P}^{M}(\Delta)}.  \end{align*}
	
	Since we assumed $\blp(I, \Delta)\leq u$, we obtain
  $\min_{\multiset{D}{m}}\phi^{\mathcal{P}^{M}(\Delta)}\leq u$.
  Moreover, since $\g$ has fully symmetric fractional polymorphisms of all
  arities, Lemma \ref{lemma:fhom} implies the existence of a fractional
  homomorphism $\omega$ from $\mathcal P^M(\Delta)$ to $\g$. From Proposition
  \ref{prop:frachom} it follows that \[\inf_{C}\phi^{\g}\leq
  \min_{\multiset{D}{m}}\phi^{\mathcal{P}^{M}(\Delta)}\leq \blp(I, \Delta)\leq u.\] 
\end{proof}

\end{document}